\documentclass[11pt,fleqn]{article}

\usepackage{amsmath,amssymb}
\usepackage{url}
\usepackage{psfig2e}
\usepackage{rotate}

\def\etal{{et al.}}

\newtheorem{axiom}{Axiom}[section]

\newtheorem{lemma}[axiom]{Lemma}
\newtheorem{observation}[axiom]{Observation}

\newtheorem{theorem}[axiom]{Theorem}

\newenvironment{proof}{{\sl Proof.\ }\rm }{}
\def\qed{\hfill $\Box$ \bigbreak}

\newcommand{\seq}[1]{\langle #1\rangle}

\newcommand{\HB}{{\sc Honey-Bee}}
\newcommand{\HBS}{{\sc Honey-Bee-Solitaire}}
\newcommand{\HBtwo}{{\sc Honey-Bee-2-Players}}
\newcommand{\SCS}{{\sc SCS}}
\newcommand{\MSCS}{{\sc MSCS}}
\newcommand{\FVS}{{\sc FVS}}
\newcommand{\QBF}{{\sc QBF}}

\long\def\probl#1#2#3{\bigskip
\centerline{\framebox{
\begin{minipage}{0.92\textwidth}
\centerline{\bf Problem #1}
\smallskip
{\bf Input:} #2
\\[1ex]
{\bf Question:} #3
\end{minipage}}}\bigskip}

\begin{document} 

\title{{\bf An Algorithmic Analysis of the\\Honey-Bee Game}\thanks{RF acknowledges
support by
the National Natural Science Foundation of China (No. 60973026),
the Shanghai Leading Academic Discipline Project (project number B114),
the Shanghai Committee of Science and Technology of China (09DZ2272800),
and the Robert Bosch Foundation (Science Bridge China 32.5.8003.0040.0).
GJW acknowledges support by
the Netherlands Organisation for Scientific Research (NWO grant 639.033.403),
and by BSIK grant 03018 (BRICKS: Basic Research in Informatics for Creating
the Know\-ledge Society).}}

\author{\sc Rudolf Fleischer
      \thanks{ School of Computer Science, IIPL, Fudan University,
      Shanghai 200433, China. Email: {\tt rudolf@fudan.edu.cn}.}
      \and
      \sc Gerhard J.\ Woeginger
      \thanks{ {\tt gwoegi@win.tue.nl}.
      Department of Mathematics and Computer Science,
      TU Eindhoven, P.O.\ Box 513, 5600 MB Eindhoven, Netherlands.}
      }
\date{}
\maketitle

\date{}

\maketitle


\begin{abstract}
The {\HB} game is a two-player board game that is played on a connected 
hexagonal colored grid or (in a generalized setting) on a connected graph 
with colored nodes.
In a single move, a player calls a color and thereby conquers all the nodes 
of that color that are adjacent to his own current territory.
Both players want to conquer the majority of the nodes.
We show that winning the game is PSPACE-hard in general, NP-hard on 
series-parallel graphs, but easy on outerplanar graphs.

In the solitaire version, the goal of the single player is to conquer the 
entire graph with the minimum number of moves. 
The solitaire version is NP-hard on trees and split graphs, but can be solved 
in polynomial time on co-comparability graphs.

\end{abstract}

\bigskip
\section{Introduction}
\label{s_intro}
The {\HB} game is a popular two-player board game that shows up in many 
different variants and at many different places on the web (the game is 
best be played on a computer). 
For a playable version we refer the reader for instance to Axel Born's 
web-page \cite{Bor09}; see Fig.~\ref{fig_born} for a screenshot.
The playing field in {\HB} is a grid of hexagonal honey-comb cells that 
come in various colors; the coloring changes from game to game.
The playing field may be arbitrarily shaped and may contain holes,
but must always be connected.
In the beginning of the game, each player controls a single cell in some
corner of the playing field.
Usually, the playing area is symmetric and the two players face each other 
from symmetrically opposing starting cells.
In every move a player may call a color $c$, and thereby gains control over 
all connected regions of color $c$ that have a common border with the area
already under his control.
The only restriction on $c$ is that it cannot be one of the two colors used 
by the two players in their last move before the current move, respectively.
A player wins when he controls the majority of all cells.
On Born's web-page \cite{Bor09} one can play against a computer,
choosing from four different layouts for the playing field.
The computer uses a simple greedy strategy: ``Always call the color $c$ that 
maximizes the immediate gain.''
This strategy is short-sighted and not very strong, and an alert human player 
usually beats the computer after a few practice matches.

\begin{figure}
\setlength{\unitlength}{1000truesp}
\begin{center}
\begin{picture}(28000,18000)%
\epsfig{file=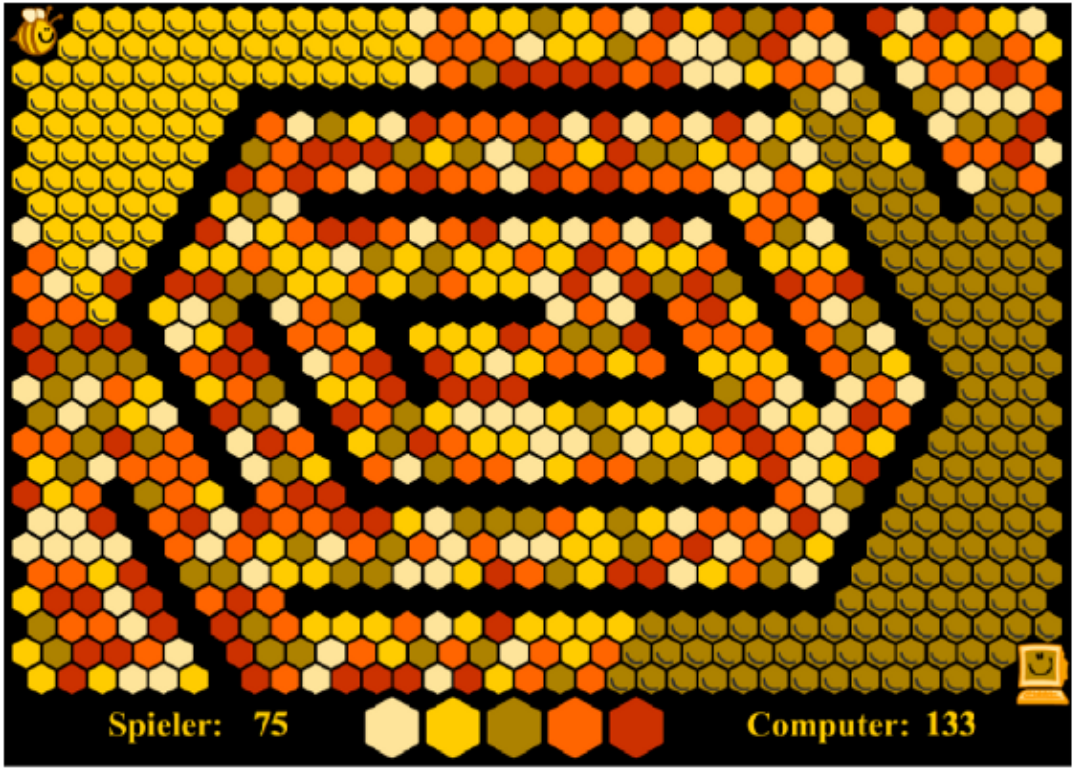}%
\end{picture}%
\caption{Born's ``Biene''. The human player (starting from the top-left corner)
is on the edge of losing against the computer (starting from the bottom-right corner).}
\label{fig_born}
\end{center}
\end{figure}

In this paper we perform a complexity study of the {\HB} game when played by two players 
on some arbitrary connected graph instead of the hex-grid of the original game.
We will show in Section~\ref{s_two} that {\HBtwo} is NP-hard even on series-parallel 
graphs, and that it is PSPACE-complete in general.
On outerplanar graphs, however, it is quite easy to compute a winning strategy.

In the \emph{solitaire} (single-player) version of {\HB} the goal is to conquer
the entire playing field as quickly as possible.
Intuitively, a good strategy for the solitaire game will be close to a
strong heuristic for the two-player game.
For the solitaire version, our results draw a sharp separation line between 
easy and difficult cases.
In particular, we show in Section~\ref{s_one} that {\HBS} is NP-hard for 
split graphs and for trees, but polynomial-time solvable on co-comparability 
graphs (which include interval graphs and permutation graphs).
Thus, the complexity of the game is well-characterized for the class and subclasses 
of perfect graphs; see Fig.~\ref{fig_results} for a summary of our results.

\begin{figure}
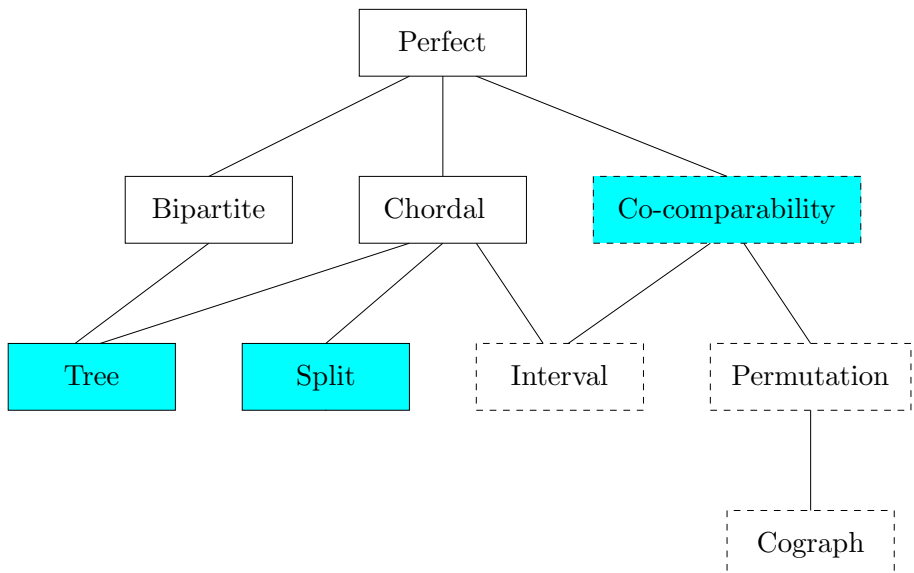

\centering
\input results.tex
\caption{Summary of the complexity results for \HBS.
NP-complete problems have a solid frame, polynomial-time solvable problems
have a dashed frame.
The results for the graph classes in the three colored boxes imply all other results.}
\label{fig_results}
\end{figure}

\section{Definitions}
\label{s_def}
We model {\HB} in the following graph-theoretic setting.
The playing field is a connected, simple, loopless, undirected graph $G=(V,E)$.
There is a set $C$ of $k$ colors, and every node $v\in V$ is colored by some color 
$col(v)\in C$; we stress that this coloring does not need to be proper, that is, 
there may be edges $[u,v]\in E$ with $col(u)=col(v)$.
For a color $c\in C$, the subset $V_c\subseteq V$ contains the nodes of color $c$.
For a node $v\in V$ and a color $c\in C$, we define the \emph{color-$c$-neighborhood} 
$\Gamma(v,c)$ as the set of nodes
in $V_c$ either adjacent to $v$ or connected to $v$ by a path
of nodes of color $c$.
Similarly, we denote by $\Gamma(W,c)=\bigcup_{w\in W}\Gamma(w,c)$
the color-$c$-neighborhood of a subset $W\subseteq V$.
For a subset $W\subseteq V$ and a sequence $\gamma=\seq{\gamma_1,\ldots,\gamma_b}$ of 
colors in $C$, we define a corresponding sequence of node sets $W_1=W$ and 
$W_{i+1}=W_i\cup \Gamma(W_i,\gamma_i)$, for $1\le i\le b$.
We say that sequence $\gamma$ started on $W$ \emph{conquers} the final node set 
$W_{b+1}$ in $b$ moves, and we denote this situation by $W\to_{\gamma}W_{b+1}$.
The nodes in $V-W_{b+1}$ are called \emph{free} nodes.

In the \emph{solitaire} version of \HB, the goal is to conquer the entire playing 
field with the smallest possible number of moves.
Note that {\HBS} is trivial in the case of only two colors.
But as we will see in Section~\ref{s_one},
the case of three colors can already be difficult.

\probl{\HBS}
{A graph $G=(V,E)$; 
a set $C$ of $k$ colors and a coloring $col:V\to C$;
a start node $v_0\in V$;
and a bound $b$.}
{Does there exist a color sequence $\gamma=\seq{\gamma_1,\ldots,\gamma_b}$
of length $b$ such that $\{v_0\}\to_{\gamma}V$?}

In the \emph{two-player} version of {\HB}, the two players $A$ and $B$ start from 
two distinct nodes $a_0$ and $b_0$ and then extend their regions step by step by 
alternately calling colors.
Player $A$ makes the first move.
One round of the game consists of a move of $A$ followed by a move of $B$.
Consider a round, where at the beginning the two players control node sets $W_A$ 
and $W_B$, respectively.
If player $A$ calls color $c$, then he extends his region $W_A$ to 
$W^\prime_A=W_A\cup(\Gamma(W_A,c)-W_B)$.
If afterwards player $B$ calls color $d$, then he extends his region $W_B$ to
$W^\prime_B=W_B\cup(\Gamma(W_B,c)-W^\prime_A)$.
Note that once a player controls a node, he can never lose it again.

The game terminates as soon as one player controls more than half of all nodes.
This player wins the game.
To avoid draws, we require that the number of nodes is odd.
There are three important rules that constrain the colors that a player is 
allowed to call.

\begin{enumerate}
\item[R1.]
A player must never call the color that has just been called by the other player.
\item[R2.]
A player must never call the color that he has called in his previous move.
\item[R3.]
A player must always call a color that strictly enlarges his territory, unless 
rules R1 and R2 prevent him from doing so.
\end{enumerate} 


\begin{figure}
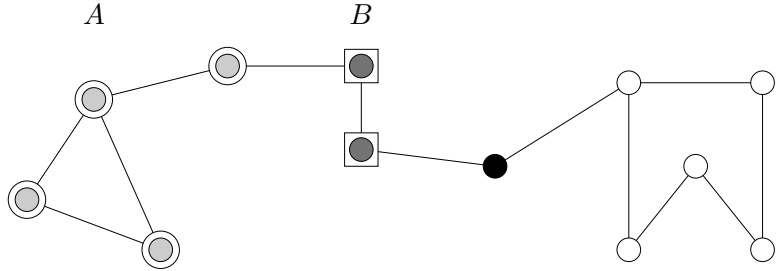

\centering
\input r2.tex
\caption{Player A (circled nodes) is leading with four captured nodes over player B
(squared nodes) with only two captured nodes.
Player B would next like to play black to capture all the white nodes
in the next move.
Without rule~R2, player A could prevent this by repeatedly
playing black.
}
\label{fig_rule_R2}
\end{figure}


What is the motivation for these three rules?
Rule~R1 is a technical condition that arises from the graphical implementation~\cite{Bor09}
of the game:
Whenever a player calls a color $c$, his current territory is entirely recolored to 
color $c$.
This makes it visually easier to recognize the territories controlled by both players.
Rule~R2 prevents the players from permanently blocking some color for the opponent.
Fig.~\ref{fig_rule_R2} shows a situation where rule~R2 actually prevents the game 
from stalling.
Rule~R3 is quite delicate, and is justified by situations as depicted in 
Fig.~\ref{fig_rule_R3}.
Rule~R3 guarantees that every game must terminate with either a win for player~A 
or a win for player~B.
Note that rule~R2 is redundant except in the case when
a player has no move to gain territory (see Fig.~\ref{fig_rule_R2}.


\begin{figure}
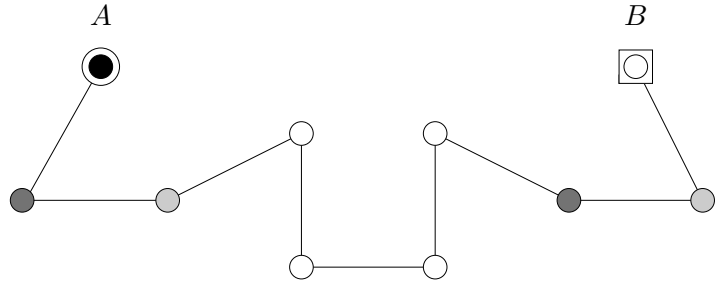

\centering
\input r3.tex
\caption{
Player A who controls the black node at the left end of the path
loses if he calls dark-gray (and hence prefers to call white, light-gray, and black).
Player B who controls the white node at the other end of the path
loses if he calls light-gray (and hence prefers to call colors white, dark-gray, and black).
Rule~R3 forces the players to move into the unoccupied territory.
}
\label{fig_rule_R3}
\end{figure}


\probl{\HBtwo}
{A graph $G=(V,E)$ with an odd number of nodes;
a set $C$ of colors and a coloring $col:V\to C$;
two start nodes $a_0,b_0\in V$.}
{Can player $A$ enforce a win when the game is played according to the
above rules?}

Note that {\HBtwo} is trivial in the case of only three colors:
The players do not have the slightest freedom in choosing their next color, and 
always must call the unique color allowed by rules~R1 and~R2.
However we will see in Section~\ref{s_two} that the case of four colors can 
already be difficult.

Finally we observe that calling a color $c$ always conquers all connected components 
induced by $V_c$ that are adjacent to the current territory.
Hence an equivalent definition of the game could use a graph with node weights 
(that specify the size of the corresponding connected component) and a \emph{proper} 
coloring of the nodes.
Any instance under the original definition can be transformed into an equivalent 
instance under the new definition by contracting each connected component of $V_c$, 
for some $c$, into a single node of weight $|V_c|$.
However, we are interested in restrictions of the game to particular graph classes,
some of which are not closed under edge contractions (as for instance the hex-grid 
graph of the original {\HB} game).

\section{The Solitaire Game}
\label{s_one}
In this section we study the complexity of finding optimally short color 
sequences for {\HBS}.
We will show that this is easy for co-comparability graphs, while it is 
NP-hard for trees and split graphs.
Since the family of co-comparability graphs contains interval graphs, 
permutation graphs, and co-graphs as sub-families, our positive result
for co-comparability graphs implies all other positive results in
Fig.~\ref{fig_results}.

A first straightforward observation is that {\HBS} lies in NP:
Any connected graph $G=(V,E)$ can be conquered in at most $|V|$ moves,
and hence such a sequence of polynomially many moves can serve as an
NP-certificate.

\subsection{The Solitaire Game on Co-Comparability Graphs}
\label{sec1:cocomparability}
A \emph{co-comparability graph} $G=(V,E)$ is an undirected graph
whose nodes $V$correspond to the elements of some partial order $<$
and whose edges $E$ connect any two elements that are incomparable in that partial order,
i.e., $[u,v]\in E$ if neither $u<v$ nor $v<u$ holds.
For simplicity, we identify the nodes with the elements of the partial order.
Golumbic {\etal}~\cite{GoRoUr83} showed that co-comparability
graphs are exactly the intersection graphs of
continuous real-valued functions over some interval $I$.
If two function curves intersect,
the corresponding elements are incomparable in the partial order;
otherwise, the curve that lies complete above the other one
corresponds to the larger element in the partial order.
The function graph representation readily implies
that the class of co-comparability graphs is closed under edge contractions.
Therefore, we may w.l.o.g.~restrict our analysis of {\HBS} to
co-comparability graphs with a proper node coloring,
i.e., adjacent nodes have distinct colors
(in the solitaire game
we do not care about the weight of a node after an edge contraction).
In this case, every color class is totally ordered
because incomparable node pairs have been contracted.

Consider an instance of {\HBS} with 
a minimal start node $v_0$ (in the partial order on $V$);
a maximal start node could be handled similarly.
The function graph representation implies the following observation.

\begin{observation}
\label{obs_smaller}
Conquering a node will simultaneously conquer all smaller nodes of the same color.
\qed
\end{observation}

For any color $c$, let $Max(c)$ denote the largest node of color $c$.
By Obs.~\ref{obs_smaller}, it suffices to find the shortest 
color sequence conquering all nodes $Max(c)$, for all colors $c$.
We can do that by a simple shortest path computation.
We assign every node $Max(c)$ weight $0$, and all other nodes weight $1$.
Then we compute a shortest path (with respect to the node-weights)
from $v_0$ to every node $Max(c)$ that is a \emph{maximal element}
in the partial order
(which is actually exactly the set of all maximal elements).
Let $OPT$ denote the smallest cost over all such paths.

For a color sequence $\gamma=\seq{\gamma_1,\ldots,\gamma_b}$,
we define the \emph{length} of $\gamma$ as $|\gamma|=b$.
We also define the \emph{essential length} $ess(\gamma)$ of $\gamma$ as
$|\gamma|$ minus the number of steps where $\gamma$
conquers a maximal node $Max(c)$ of some color class $c$.
Obviously, $|\gamma|=ess(\gamma)+k$.
Note that $OPT$ is the minimal essential cost of any
color sequence conquering one of the maximal nodes.

\begin{lemma}
\label{thm_opt}
The optimal solution for {\HBS} has cost $OPT+k$.
\end{lemma}

\begin{proof}
Let $\gamma$ be a shortest color sequence conquering the entire graph
starting at $v_0$.
After conquering $v$, $\gamma$ only needs to conquer all free nodes $Max(c)$
to conquer the entire graph.
Thus, $|\gamma| = ess(\gamma) + k \ge OPT+k$.
\qed
\end{proof}

\begin{theorem}
\label{th_cocomp}
{\HBS} starting at an extremal node $v_0$
can be solved in polynomial time on co-comparability graphs.
\end{theorem}

\begin{proof}
Given the co-comparability graph $G$, we can compute the underlying
partial order $<$ in polynomial time \cite{GoRoUr83}.
Assigning the weights and solving one single source shortest path problem
starting at $v_0$ also takes polynomial time.
\qed
\end{proof}

We can also formulate this algorithm as a dynamic program.
For any node $v$, let $D(v)$ denote the essential length of the shortest
color sequence $\gamma$ that can conquer $v$ when starting at $v_0$.
For any color $c$, let $min_v(c)$ denote the smallest node of color $c$
connected to $v$, if such nodes exist.
Then we can compute $D(v)$ recursively as follows:
$$D(v_0) = 0$$
and
$$D(v) = \min_{c} (D(min_v(c)) + \delta_v) \>,$$
where
$D(min_v(c))=\infty$ if $min_v(c)$ is undefined, and
$\delta_v=0$ ($1$) if $v$ is (not) a maximal node for some color class.

Clearly, this dynamic program simulates the shortest path computation
of our first algorithm and we have $OPT = \min_{v}(D(v)+k)$,
where we minimize over all maximal nodes $v$. 
We now extend the dynamic program to the case that $v_0$ is not an extremal element.
The problem is that we now must extend our territory in two directions.
If we choose a move that makes good progress upwards it may make little progress downwards,
or vice versa.
In particular, the optimal strategy cannot be decomposed into two independent
optimal strategies, one conquering upwards and one conquering downwards.
Analogously to the algorithm above, for a clor $c$ define
$Min(c)$ as the smallest node of color $c$,
and $max_v(c)$ as the largest node of color $c$ connected to a node $v$.

Unfortunately, we must now redefine the essential length of a color sequence $\gamma$.
In our original definition, we did not count coloring steps
that conquered maximal elements of some color class.
This is intuitively justified by the fact that these steps must be
done by any color sequence conquering the entire graph at some time,
therefore it is advantageous to do them as early as possible
(which is guaranteed by giving these moves cost 0).
But now we must also consider the minimal nodes of each color class.
An optimal sequence conquering the entire graph will at some time have conquered
a minimal node and a maximal node.
Afterwards, it will only call extremal nodes for some color class.
If both extremal nodes of a color class are still free,
we only need \emph{one} move to conquer both simultaneously.
If one of them had been captured earlier, we still need to conquer the other one.
This indicates that we should charge 1 for the first extremal node conquered
while the second one should be charged 0, as before. 
If both nodes are conquered in the same move, we should also charge 0.
Therefore, we now define the \emph{essential length} $ess(\gamma)$ of $\gamma$ as
$|\gamma|$ minus the number of steps where $\gamma$
conquers the second extremal node of some color class.

For a node $v$ below $v_0$ or incomparable to $v_0$
and a node $w$ above $v_0$ or incomparable to $v_0$ let
$D(v,w)$ denote the essential length of the shortest
color sequence $\gamma$ that can conquer $v$ and $w$ when starting at $v_0$.
Note that we do not need to keep track of which first extremal nodes
of a color class have been conquered because we can deduce this from the
two nodes $v$ and $w$ currently under consideration.
In particular, we can compute $D(v,w)$ recursively as follows:
$$D(v_0,v_0) = 0$$ 
and 
$$D(v,w) = \min_{c} (D(v,min_w(c)) + \delta_w(v),
D(max_v(c),w) + \delta_v(w)) \>,$$
where $\delta_v(w)=0$ if and only if $w$ is an extremal node of some
color class $c$ and the other extremal node of color class $c$
is either between $v$ and $w$, or incomparable to either $v$ or $w$, or both
(it was either conquered earlier, or it will be conquered in this step);
otherwise, $\delta_v(w)=1$.
Obviously, $|\gamma|=ess(\gamma)+k$.

\begin{lemma}
\label{thm_opt_general}
The optimal solution for {\HBS} has cost $\min_{v,w}(D(v,w)+k)$,
where we minimize over all minimal nodes $v$ and all maximal nodes $w$.
\end{lemma}

\begin{proof}
Let $\gamma$ be a shortest color sequence conquering the entire graph
starting at $v_0$.
Let $v$ be the first minimal node conquered by $\gamma$
and $w$ the first maximal node.
After conquering $v$ and $w$, $\gamma$ only needs to conquer all free nodes $Max(c)$
to conquer the entire graph.
Thus, $|\gamma| \ge D(v,w) + k$.
\qed
\end{proof}

\begin{theorem}
\label{th_cocomp_general}
{\HBS} can be solved in polynomial time on co-comparability graphs.
\qed
\end{theorem}

\subsection{The Solitaire Game on Split Graphs}
\label{ss_split}
A \emph{split graph} is a graph whose node set can be partitioned into an induced 
clique and into an induced independent set.
We will show that {\HBS} is NP-hard on split graphs.
Our reduction is from the NP-hard {\tt Feedback Vertex Set} (\FVS)  problem in 
directed graphs; see for instance Garey and Johnson \cite{GaJo79}.

\probl{\FVS}
{A directed graph $(X,A)$; a bound $t<|X|$.}
{Does there exist a subset $X^\prime\subseteq X$ with $|X^\prime|=t$ such that the
directed graph induced by $X-X^\prime$ is acyclic?}

\begin{theorem}
\label{thm_split}
{\HBS} on split graphs is NP-hard.
\end{theorem}

\begin{proof}
Consider an instance $(X,A,t)$ of \FVS.
To construct an instance $(V,E,b)$ of \HBS,
we first build a clique from the nodes in $X$ together with a new
node $v_0$, the start node of \HBS,
where each node $x\in X+v_0$ has a different color $c_x$.
Next, we build the independent set.
For every arc $(x,y)\in A$, we introduce a corresponding node $v(x,y)$ of color 
$c_y$ which is only connected to node $x$ in the clique, i.e., it has degree one.
Finally, we set $b=|X|+t$.
We claim that the constructed instance of {\HBS} has answer YES,
if and only if the instance of {\FVS} has answer YES.

Assume that the {\FVS} instance has answer YES.
Let $X^\prime$ be a smallest feedback set whose removal makes $(X,A)$ acyclic.
Let $\pi$ be a topological order of the nodes in $X-X^\prime$,
and let $\tau$ be an arbitrary ordering of the nodes in $X^\prime$.
Consider the color sequence $\gamma$ of length $|X|+t$ that starts with $\tau$, 
followed by $\pi$, and followed by $\tau$ again.
We claim that $\{v_0\}\to_{\gamma}V$.
Indeed, $\gamma$ first runs through $\tau$ and $\pi$ and thereby conquers
all clique nodes.
Every independent set node $v(x,y)$ with $y\in X^\prime$ is conquered during the second
transversal of $\tau$.
Every independent set node $v(x,y)$ with $y\in X-X^\prime$ is conquered during the
transversal of $\pi$, since $\pi$ first conquers $x$ with color $c_x$, and afterwards
$v(x,y)$ with color $y$.

Next assume that the instance of {\HBS} has answer YES.
Let $\gamma$ be a color sequence of length $b=|X|+t$ conquering $V$.
Define $X^\prime$ as the set of nodes $x$ such that
color $c_x$ occurs at least twice in $\gamma$;
clearly, $|X^\prime|\le t$.
Consider an arc $(x,y)\in A$ with $x,y\in X-X^\prime$.
Since $\gamma$ contains color $c_y$ only once, it must conquer
node $v(x,y)$ of color $c_y$ after node $v(x)$ of color $c_x$.
Hence, $\gamma$ induces a topological order of $X-X^\prime$.
\qed
\end{proof}

The construction in the proof above uses linearly many colors.
What about the case of few colors?
On split graphs, {\HBS} can always be solved by traversing the color
set $C$ twice; the first traversal conquers all clique nodes, and the second traversal
conquers all remaining free independent set nodes.
Thus, every split graph can be completely conquered in at most $2|C|$ steps.
If there are only few colors, we can simply check all color sequences of this 
length $2|C|$.

\begin{theorem}
\label{thm_split_const}
If the number of colors is bounded by a fixed constant, {\HBS} on 
split graphs is polynomial-time solvable.
\qed
\end{theorem}

\subsection{The Solitaire Game on Trees}
\label{ss_tree}
In this section we will show that {\HBS} is NP-hard on trees,
even if there are at only three colors.
We reduce {\HBS} from a variant
of the {\tt Shortest Common Supersequence} (\SCS) problem
which is know to be NP-complete (see Middendorf~\cite{Mid94}).

\probl{\SCS}
{A positive integer $t$;
finite sequences $\sigma_1,\ldots,\sigma_s$ with elements from $\{0,1\}$ with
the following properties: 
(i) All sequences have the same length.
(ii) Every sequence contains exactly two 1s, and these two 1s are separated by
at least one 0.}
{Does there exist a sequence $\sigma$ of length $t$ that contains 
$\sigma_1,\ldots,\sigma_s$ as subsequences?}

Middendorf's hardness result also implies the hardness of the following variant 
of \SCS:

\probl{{\tt Modified SCS} (\MSCS)}
{A positive integer $t$;
finite sequences $\sigma,\ldots,\sigma_s$ with elements from $\{0,1,2\}$ with
the following property:
In every sequence any two consecutive elements are distinct,
and no sequence starts with 2.}
{Does there exist a sequence $\sigma$ of length $t$ that contains
$\sigma_1,\ldots,\sigma_s$ as subsequences?}

\begin{theorem}
\label{thm_mscs}
{\MSCS} is NP-complete.
\end{theorem}

\begin{proof}
Here is a reduction from {\SCS} to {\MSCS}.
Consider an arbitrary sequence $\tau$ with elements from $\{0,1\}$.
We define $f(\tau)$ as the sequence we obtain from replacing
every occurrence of the element 0 in $\tau$ by two consecutive elements 0 and 2.
Now consider an instance $(\sigma_1,\ldots,\sigma_s,t)$ of \SCS.
We construct an instance $(\sigma_1^\prime,\ldots,\sigma_s^\prime,t^\prime)$
of {\MSCS} by setting $\sigma^\prime_i=f(\sigma_i)$, for $1\le i\le s$.
Then, for any sequence $\sigma$ with elements from $\{0,1\}$,
$\sigma$ is a common supersequence of $\sigma_1,\ldots,\sigma_s$ if and only if
$f(\sigma)$ is a common supersequence of $\sigma^\prime_1,\ldots,\sigma^\prime_s$.
This implies the NP-hardness of \MSCS.
\qed
\end{proof}

\begin{theorem}
\label{thm_tree}
{\HBS} is NP-hard on trees, even in case of only three colors.
\end{theorem}

\begin{proof}
We reduce {\MSCS} to {\HBS} on trees. 
Consider an instance $(\sigma_1,\ldots,\sigma_s,t)$ of \MSCS.
We use color set $C=\{0,1,2\}$.
We first construct a root $v_0$ of color $2$.
Then we attach a path of length $|\sigma_i|$ to $v_0$
for each sequence $\sigma_i$, where an element $j$ is colored $j$.
See the left half of Fig.~\ref{fig_sp} for an example.
Finally, we set $b=t$.
It its straightforward to see that the constructed instance of {\HBS} has 
answer YES if and only if the instance of {\MSCS} has answer YES.
\qed
\end{proof}

\section{The Two-Player Game}
\label{s_two}
In this section we study the complexity of the two-player game.
While on outerplanar graphs the players can compute their winning strategies 
in polynomial time, this problem is NP-hard for series-parallel graphs with 
four colors, and PSPACE-complete with four colors on arbitrary graphs.
Our positive result for outerplanar graphs works for an arbitrary number of
colors.
Our negative results work for four colors, which is the strongest possible type
of result (recall that instances with three colors are trivial to solve).

\subsection{The Two-Player Game on Outer-Planar Graphs}
\label{ss_outerplanar}
A graph is \emph{outer-planar} if it contains neither $K_4$ nor $K_{2,3}$ as a minor.
Outer-planar graphs have a planar embedding in which every node lies on the boundary
of the so-called \emph{outer face}.
For example, every tree is an outer-planar graph.

Consider an outer-planar graph $G=(V,E)$ as an instance of {\HBtwo} with starting 
nodes $a_0$ and $b_0$ in $V$, respectively.
The starting nodes divide the nodes on the boundary of the outer face $F$ into an 
upper chain $u_1,\ldots,u_s$ and a lower chain $\ell_1,\ldots,\ell_t$,
where $u_1$ and $\ell_1$ are the two neighbors of $a_0$ on $F$,
while $u_s$ and $\ell_t$ are the two neighbors of $b_0$ on $F$.
We stress that this upper and lower chain are not necessarily disjoint 
(for instance, articulation nodes will occur in both chains).

Now consider an arbitrary situation in the middle of the game.
Let $U$ (respectively $L$) denote the largest index $k$ such that player $A$ has 
conquered node $u_k$ (respectively node $\ell_k$).
See Fig.~\ref{fig_outerplanar} to illustrate these definitions and
the following lemma.

\begin{lemma}
\label{thm_outerplanar_conquer}
Let $X$ denote the set of nodes among $u_1,\ldots,u_U$ and $\ell_1,\ldots,\ell_L$
that currently do neither belong to $A$ nor to $B$.
Then no node in $X$ can have a neighbor among 
$u_{U+1},\ldots,u_s,b_0,\ell_t,\ldots,\ell_{L+1}$.
\end{lemma}

\begin{proof}
The existence of such a node in $X$ would lead to a $K_4$-minor in the outer-planar graph.
\qed
\end{proof}

\begin{figure}
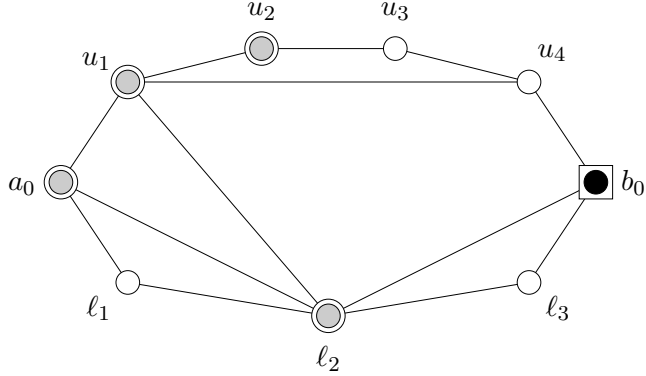

\centering
\input outerplanar.tex
\caption{An outerplanar graph with start nodes $a_0$ and $b_0$.
Player $A$ (circled nodes) has conquered the light-gray colored nodes, i.e., $U=2$ and $L=2$.
Eventually, $A$ will also conquer $\ell_1$, since Player $B$ cannot reach it.}
\label{fig_outerplanar}
\end{figure}

\begin{theorem}
\label{thm_outerplanar}
{\HBtwo} on outer-planar graphs is polynomial-time solvable.
\end{theorem}

\begin{proof}
The two indices $U$ and $L$ encode all necessary information on the future 
behavior of player $A$.
Eventually, he will own all nodes $u_1,\ldots,u_U$ and $\ell_1,\ldots,\ell_L$,
and the possible future expansions of his area beyond $u_U$ and $\ell_L$
only depend on $U$ and $L$.
Symmetric observations hold true for player $B$.

As every game situation can be concisely described by just four indices,
there is only a polynomial number $O(|V|^4)$ of relevant game situations.
The rest is routine work in combinatorial game theory:
We first determine the winner for every end-situation, and then by working backwards 
in time we can determine the winners for the remaining game situations.
\qed
\end{proof}

\subsection{The Two-Player Game on Series-Parallel Graphs}
\label{ss_sp}
A graph is \emph{series-parallel} if it does not contain $K_4$ as a minor.
Equivalently, a series-parallel graph can be constructed from a single 
edge by repeatedly doubling edges, or removing edges, or replacing edges 
by a path of two edges with a new node in the middle of the path.
We stress that we do not know whether the two-player game on series-parallel 
graphs is contained in the class NP (and we actually see no reason why it 
should lie in NP); therefore the following theorem only states NP-hardness.

\begin{theorem}
\label{thm_sp}
For four (or more) colors, problem {\HBtwo} on series-parallel graphs is NP-hard.
\end{theorem}

\begin{proof}
We use the color set $C=\{0,1,2,3\}$.
A central feature of our construction is that player $B$ will have no real decision 
power, but will only follow the moves of player $A$:
If player $A$ starts a round by calling color $0$ or $1$, then player $B$ must follow
by calling the other color in $\{0,1\}$ (or waste his move).
And if player $A$ starts a round by calling color $2$ or $3$, then player $B$ must 
call the other color in $\{2,3\}$ (or waste his move).
In the even rounds the players will call the colors in $\{0,1\}$ and in the odd 
rounds they will call the colors in $\{2,3\}$.
Both players are competing for a set of honey pots in the middle of the battlefield,
and need to get there as quickly as possible.
If a player deviates from the even-odd pattern indicated above, he might perhaps 
waste his move and delay the game by one round (in which neither player comes closer 
to the honey pots), but this remains without further impact on the outcome of the game.

The proof is by reduction from the supersequence problem {\SCS} with binary sequences;
see Section~\ref{ss_tree}.
Consider an instance $(\sigma_1,\ldots,\sigma_s,t)$ of {\SCS}, and let $n$ denote the 
common length of all sequences $\sigma_i$.
We first construct two start nodes $a_0$ and $b_0$ of colors $2$ and $3$, respectively.
For each sequence $\sigma_i$ with $1\le i\le s$ we do the following:

\begin{itemize}
\item
We construct a path $P_i$ that consists of $2n-1$ nodes and that is attached to $a_0$:
The $n$ nodes with odd numbers mimic sequence $\sigma_i$,
while the $n-1$ nodes with even numbers along the path all receive color $2$.
The first node of $P_i$ is adjacent to $a_0$, and its last node is connected to 
a so-called honey pot $H_i$. 

\item
The \emph{honey pot} $H_i$ is a long path consisting of $4st$ nodes of color $3$.
Intuitively, we may think of a honey pot as a single node of large weight, because 
conquering one of the nodes will simultaneously conquer the entire path.

\item
Every honey pot $H_i$ can also be reached from $b_0$ by another path $Q_i$ that 
consists of $2t-1$ nodes.
Nodes with odd numbers get color $0$, and nodes with even numbers get color $3$.
The first node of $Q_i$ is adjacent to $b_0$, and its last node is connected 
to $H_i$.
Furthermore, we create for each odd-numbered node (of color $0$) a new twin node 
of color $1$ that has the same two neighbors as the color $0$ node.
Note that for every path $Q_i$ there are $t$ twin pairs.
\end{itemize}

Finally we create a private honey pot $H_B$ for player $B$, that is connected 
to node $b_0$ and that consists of $4s(s-1)t+(2n-1)s$ nodes of color $2$.
This completes the construction; see Fig.~\ref{fig_sp} for an example.

Assume that the {\SCS} instance has answer YES. 
During his first $2t-1$ steps, player $B$ can only conquer the paths $Q_i$ and his
private honey pot $H_B$.
At the same time, player $A$ can conquer all paths $P_i$ by calling color $2$ in his 
even moves and by following a shortest 0-1 supersequence in his odd moves.
Then, in round $2t$ player $A$ will simultaneously conquer all the honey pots $H_i$
with $1\le i\le s$.
This gives $A$ a territory of at least $1+(2n-1)s+4s^2t$ nodes,
and $B$ a smaller territory of at most $1+(3t-1)s+4s(s-1)t+(2n-1)s$ nodes.
Hence $A$ can enforce a win.

Next assume that player $A$ has a winning strategy.
Player $B$ can always conquer his starting node $b_0$ and his private honey pot $H_B$.
If $B$ also manages to conquer one of the pots $H_i$, then he gets a territory of at
least $1+4s(s-1)t+(2n-1)s+4st$ nodes and surely wins the game.
Hence player $A$ can only win if he conquers all $s$ honey pots $H_i$.
To reach them before player $B$ does, player $A$ must conquer them within his 
first $2t$ moves.
In every odd round, player $A$ will call a color $0$ or $1$ and player $B$ will call
the other color in $\{0,1\}$.
Hence, in the even rounds, colors $0$ and $1$ are forbidden for player $A$, and the only
reasonable move is to call color $2$.
Note that the slightest deviation of these forced moves would give player $B$ a deadly 
advantage.
In order to win, the odd moves of player $A$ must induce a supersequence of length
at most $t$ for all sequences $\sigma_i$.
Therefore, the {\SCS} instance has answer YES.
\qed
\end{proof}

\begin{figure}
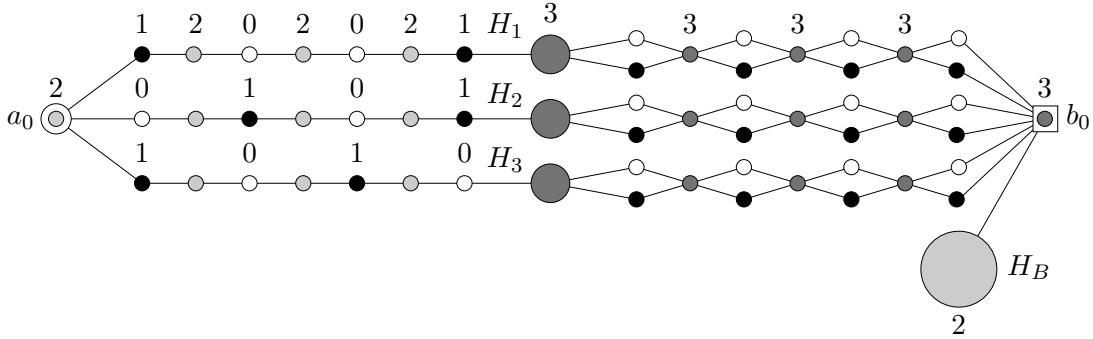

\centering
\input sp.tex
\caption{The graph constructed in the proof of Thm.~\ref{thm_sp} for the 
sequences $\sigma_1=1001$, $\sigma_2=0101$, $\sigma_3=1010$, and $t=4$.
The optimal {\SCS} solution is $10101$.
Thus, $B$ can win this game.}
\label{fig_sp}
\end{figure}

\subsection{The Two-Player Game on Arbitrary Graphs}
\label{ss_pspace}
In this section we will show that problem {\HBtwo} is PSPACE-complete
on arbitrary graphs.
Our reduction is from the PSPACE-complete {\tt Quantified Boolean Formula} (\QBF) problem;
see for instance Garey \& Johnson \cite{GaJo79}.

\probl{\QBF}
{A quantified Boolean formula with $2n$ variables
in conjunctive normal form:
$\exists x_1\forall x_2\cdots\exists x_{2n-1}\forall x_{2n} \wedge_j C_j$,
where the $C_j$ are clauses of the form $\vee_k l_{jk}$,
where the $l_{jk}$ are literals.}
{Is the formula true?}

\begin{theorem}
\label{thm_pspace}
For four (or more) colors, problem {\HBtwo}  on arbitrary graphs is PSPACE-complete.
\end{theorem}

\begin{proof}
We reduce from {\QBF}.
Let $F=\exists x_1\forall x_2\cdots\exists x_{2n-1}\forall x_{2n}
\bigwedge_j C_j$ be an instance of \QBF.
We construct a bee graph $G_F=(V,E)$ with four colors (white, light-gray, dark-gray, and 
black) such that player $A$ has a winning strategy if and only if $F$ is true.  
Let $a_0$ (colored light-gray) and $b_0$ (colored dark-gray) denote the start nodes of 
players $A$ and $B$, respectively.

Each player controls a \emph{pseudo-path}, that is, a path where some nodes may be 
duplicated as parallel nodes in a diamond-shaped structure; see Fig.~\ref{fig_var}.
A so-called \emph{choice pair} consists of a node on a pseudo-path together with 
some duplicated node in parallel.
The start nodes are at one end of the respective pseudo-paths, and the players can 
conquer the nodes on their own path without interference from the other player.
However, they must do so in a timely manner because either path ends at a humongous
\emph{honey pot}, denoted respectively by $H_A$ and $H_B$.
A honey pot is a large clique of identically-colored nodes (we may think of it as 
a single node of large weight, because conquering one node will simultaneously 
conquer the entire clique).
Both honey pots have the same size but different colors, namely black ($H_A$)
and white ($H_B$), and they are connected to each other by an edge.
Consequently, both players must rush along their pseudo-paths as quickly as possible 
to reach their honey pot before the opponent can reach it and to prevent the opponent
from winning by conquering both honey pots.
The last nodes before the honey pots are denoted by $a_f$ and $b_f$, respectively.
They separate the last variable gadgets (described below) from the honey pots.

\begin{figure}[ht]
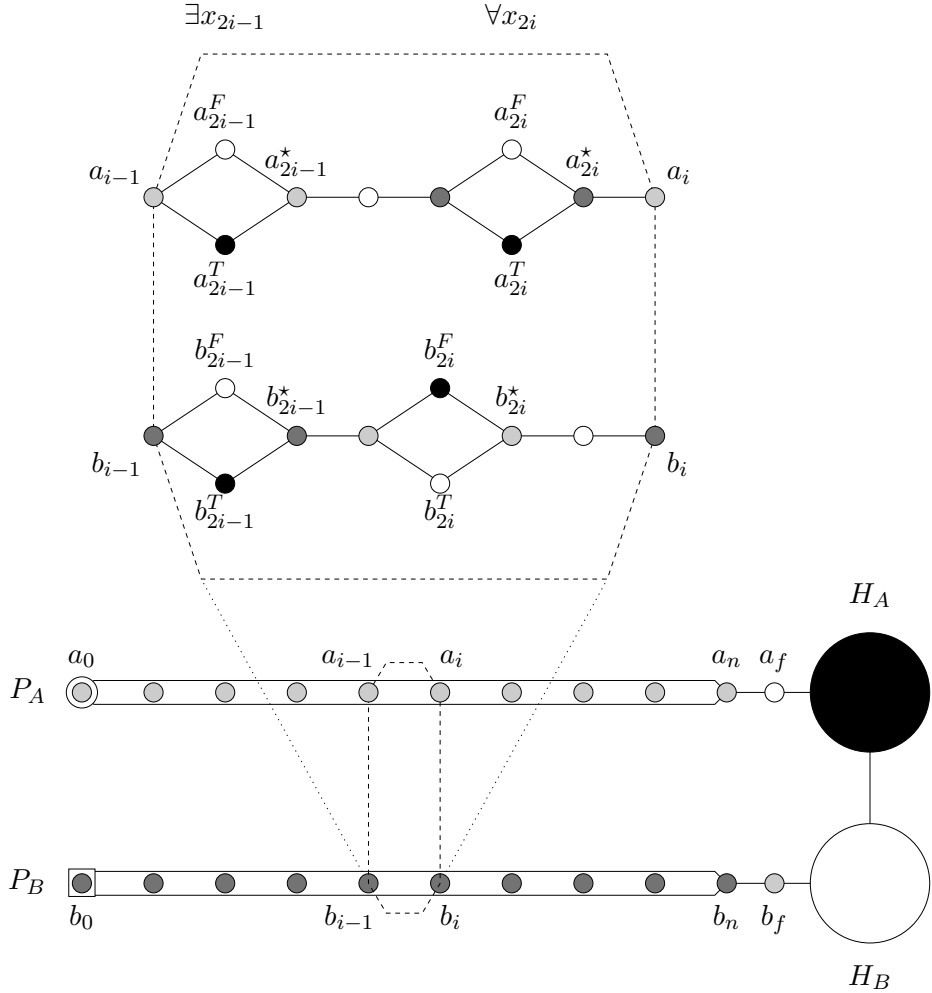

\centering
\input var.tex
\caption{The variable gadget in the proof of Thm.~\ref{thm_pspace}.}
\label{fig_var}
\end{figure}

Fig.~\ref{fig_var} shows an overview of the pseudo-paths
and one \emph{variable gadget} in detail.
A variable gadget is a part of the two pseudo-paths corresponding to
a pair of variables $\exists x_{2i-1} \forall x_{2i}$, for some $i\ge1$.
For player $A$, the gadget starts at node $a_{i-1}$ with a choice pair
$a_{2i-1}^F$ and $a_{2i-1}^T$, colored white and black, respectively.
The first node conquered by $A$ will determine
the truth value for variable $x_{2i-1}$.
In the same round, player $B$ has a choice on his pseudo-path $P_B$ between nodes
$b_{2i-1}^F$ and $b_{2i-1}^T$.
Since these nodes have the same color as $A$'s choices in the same round,
$B$ actually does not have a choice but must select the other color
not chosen by $A$.

Three rounds later, player $B$ has a choice pair
$b_{2i}^F$ and $b_{2i}^T$, assigning a truth value to variable $x_{2i}$.
In the next step (which is in the next round), player $A$ has a choice pair
$a_{2i}^F$ and $a_{2i}^T$ with the same colors as $B$'s choice pair for $x_{2i}$.
Again, this means that $A$ does not really have a choice but must select the color not chosen
by $B$ in the previous step.
Since we want $A$ to conquer those clauses containing a literal set to true by player $B$,
the colors in $B$'s choice pair have been switched, i.e., $b_{2i}^F$ is black and
$b_{2i}^T$ is white.

Note that all the nodes $a_0,a_1,\ldots,a_n$ are light-gray
and all the nodes $b_0,b_1,\ldots,b_n$ are dark-gray.
This allows us to concatenate as many variable gadgets as needed.
Further note that $a_f$ is white, while $b_f$ is light-gray.

The \emph{clause} gadgets are very simple.
Each clause $C_j$ corresponds to a small honey pot $H_j$ of color white.
The size of the small honey pots is smaller than the size of the large honey pots
$H_A$ and $H_B$, but large enough such that player $A$ loses if he misses one of
them.
Player $A$ should conquer $H_j$ if and only if $C_j$ is true in the assignment
chosen by the players while conquering their respective pseudo-paths.
We could connect $a_{2i-1}^T$ directly with $H_j$ if $C_j$ contains literal $x_{2i-1}$,
however then player $A$ could in subsequent rounds shortcut his pseudo-path
by entering variable gadgets for the other variables in $C_j$ from $H_j$.
To prevent this from happening, we place waiting gadgets
between the variable gadgets and the clauses.

Let $a_{k}^\star$ denote the node on $P_A$ right after
the choice pair $a_k^F$ and $a_k^T$, for $k=1,\ldots,2n$;
similarly, $b_k^\star$ are the nodes on $P_B$ right after $B$'s choice pairs.
A \emph{waiting gadget} $W_k$ consists of two copies $W_k^F$ and
$W_k^T$ of the sub-path of $P_A$ starting at $a_k^\star$ and ending at $a_n$,
see Fig.~\ref{fig_wait}.
If clause $C_j$ contains literal $x_k$, $H_j$ is connected to the node $w_n^T$
corresponding to $a_n$ in $W_k^T$;
if $C_j$ contains literal $\overline{x_k}$, $H_j$ is connected to the node $w_n^F$
corresponding to $a_n$ in $W_k^F$.
If $k=2i-1$ (i.e., we have an existential variable $x_{2i-1}$ whose value is assigned by
player $A$), then $a_{2i-1}^F$ and $b_{2i-1}^F$ are connected to $w_{2i-1}^{\star F}$,
and $a_{2i-1}^T$ and $b_{2i-1}^T$ are connected to $w_{2i-1}^{\star T}$.
If $k=2i$ (i.e., we have a universal variable $x_{2i}$ whose value is assigned by
player $B$), then $a_{2i}^F$ and $b_{2i}^\star$ are connected to $w_{2i}^{\star F}$,
and $a_{2i}^T$ and $b_{2i-1}^\star$ are connected to $w_{2i}^{\star T}$.

\begin{figure}
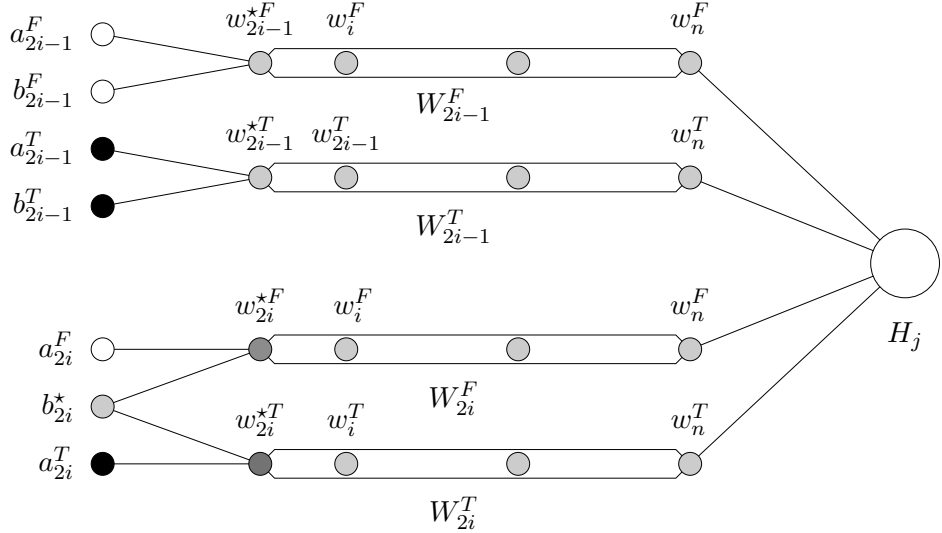

\centering
\input wait.tex
\caption{The waiting gadgets for existential variables
($W_{2i-1}^F$ and $W_{2i-1}^T$, the two top paths)
and universal variables ($W_{2i}^F$ and $W_{2i}^T$, the two bottom paths)
in the proof of Thm.~\ref{thm_pspace}.
Note that usually one of the two waiting paths $W_k^F$ or $W_k^T$ woud be connected to $H_j$
because we may assume that a clause does not contain $x_k$ and $\overline{x_k}$
at the same time.
}
\label{fig_wait}
\end{figure}

Finally, we connect $b_f$ with all clause honey pots $H_j$
to give player $B$ the opportunity to conquer all those
clauses that contain no true literal.
This completes the construction of $G_F$.
Fig.~\ref{fig_example} shows the complete graph $G_F$ for a small example formula $F$.

We claim that player $A$ has a winning strategy on $G_F$
if and only if formula $F$ is true.
It is easy to verify that player $A$ can indeed win if $F$ is true.
All he has to do is to conquer those nodes in his existential choice pairs
corresponding to the variable values in a satisfying assignment for $F$.
For the existential variables, he has full control to select any value,
and for the universal variables he must pick the opposite color as selected by player 
$B$ in the previous step, which corresponds to setting the variable to exactly the value
that player $B$ has selected.
Hence player $B$ can block a move of player $A$ by appropriately selecting a
value for a universal variable.
Note that no other blocking moves of player $B$ are advantageous:
If $B$ blocks $A$'s next move by choosing a color that does not make progress 
on his own pseudo-path, then $A$ will simply make an arbitrary waiting move and 
then in the next round $B$ cannot block $A$ again.
When player $A$ conquers node $a_n$, he will simultaneously conquer the last nodes
in all waiting gadgets corresponding to true literals.
Since every clause contains a true literal for a satisfying assignment, player $A$ 
can then in the next round conquer $a_f$ together with all clause honey pots
(which all have color white).
Player $B$ will respond by conquering $b_f$, and the game ends with both players 
conquering their own large honey pots $H_A$ and $H_B$, respectively.
Since player $A$ got all clause honey pots, he wins.

To make this argument work, we must carefully chose the sizes of the honey pots.
Each pseudo-path contains $9n+1$ nodes, of which at most $n$ can be conquered by 
the other player.
The waiting gadgets contain two paths of length $9k+6$ for existential variables
and $9k+1$ for universal variables.
At the end, player $A$ will have conquered one of the two paths completely and maybe
some parts of the sibling path, that is, we do not know exactly the final owner of
less than $n^2$ nodes.
The clause honey pots should be large enough to absorb this fuzzyness, which means 
it is sufficient to give them $2n^2$ nodes.
The honey pots $H_A$ and $H_B$ should be large enough to punish any foul play by the 
players, that is, when they do not strictly follow their pseudo-paths.
It is sufficient to give them $2n^3$ nodes.

To see that $F$ is true if player $A$ has a winning strategy note that player $A$ 
must strictly follow his pseudo-path, as otherwise player $B$ could beat him by 
reaching the large honey pots first.
Thus player $A$'s strategy induces a truth assignment for the existential variables.
Similarly, player $B$'s strategy induces a truth assignment for the universal variables.
Player $A$ can only win if he also conquers all clause honey pots, and hence the players 
must haven chosen truth values that make at least one literal per clause true.
This means that formula $F$ is satisfiable.
\qed
\end{proof}

\section{Conclusions}
\label{s_conclusion}
We have modeled the Honey Bee game as a combinatorial game on colored graphs.
For the solitaire version, we have analyzed the complexity on many classes of
perfect graphs.
For the two player version, we have shown that even the highly restricted 
case of series-parallel graphs is hard to tackle.
Our results draw a clear separating line between easy and hard variants of
these problems.

\section{Acknowledgements} 
Part of this research was done while G.~Woeginger visited Fudan University in 2009.


\bibliographystyle{plain}

\begin{figure}
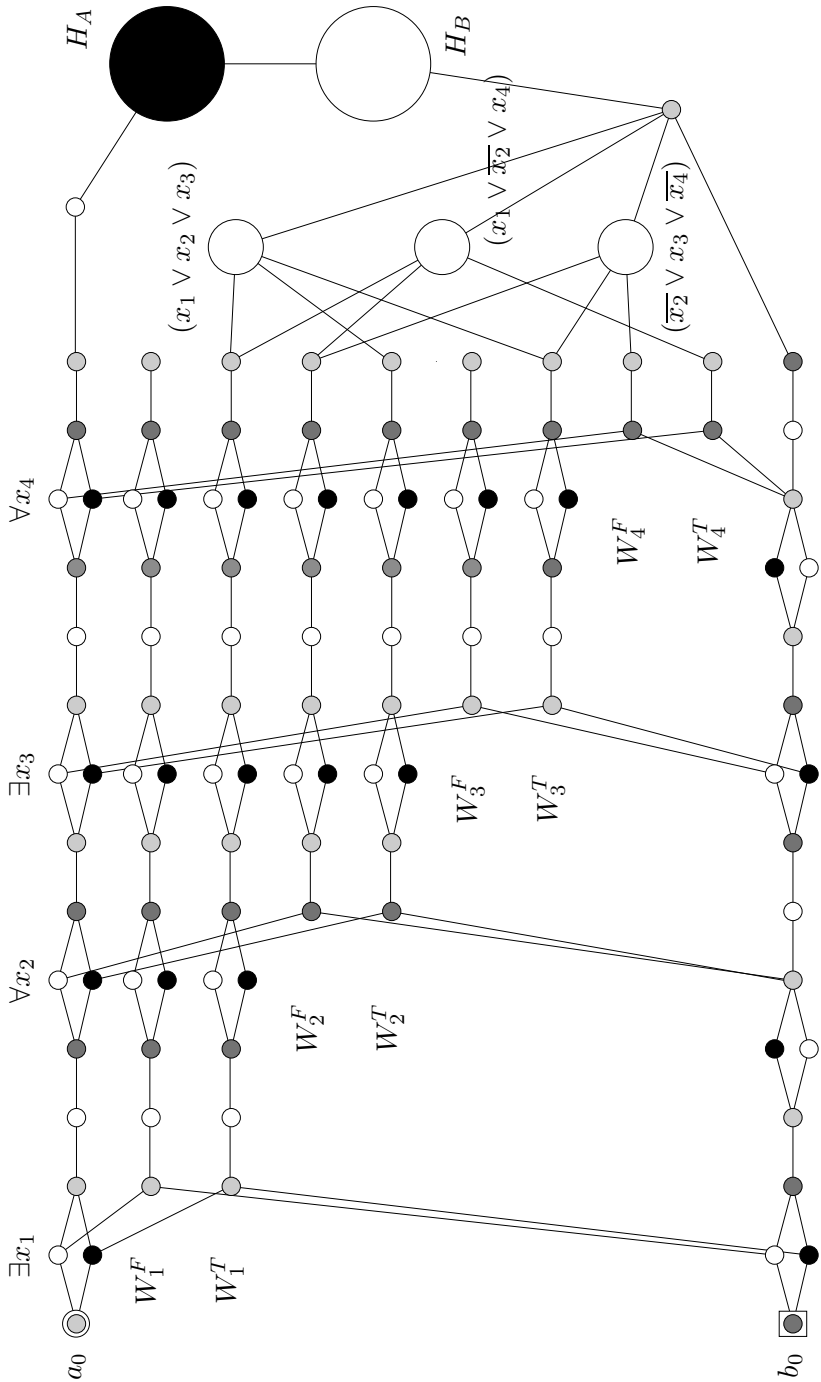

\begin{center}
\rotate{\input graph.tex }
\caption{The reduction in the proof of Thm.~\ref{thm_pspace} would produce this graph
for the formula $F=(x_1\vee x_2\vee x_3) \wedge (x_1\vee \overline{x_2}\vee x_4)
\wedge (\overline{x_2}\vee x_3\vee\overline{x_4})$.}
\label{fig_example}
\end{center}
\end{figure}

\end{document}